\documentclass[a4paper,aps,prl,twocolumn,10pt,superscriptaddress]{revtex4-1}

\usepackage[utf8]{inputenc}
\usepackage[T1]{fontenc}
\usepackage[sc,osf]{mathpazo}
\usepackage{amsmath, amsthm, amssymb}
\usepackage{graphicx}
\usepackage{dcolumn}
\usepackage{bm}
\usepackage{bbm}
\usepackage{hyperref}
\usepackage{tikz}

\usetikzlibrary{calc}
\usetikzlibrary{decorations.pathmorphing}

\newtheorem{defn}{Definition}
\newtheorem{thm}[defn]{Theorem}

\newtheorem{cor}[defn]{Corollary}

\newtheorem{lem}[defn]{Lemma}

\theoremstyle{definition}

\def\RR{\mathbbm{R}}
\def\minimize{\textrm{minimize}}
\def\st{\textrm{subject to }}

\newcommand{\op}[1]{{\mathbf{#1}}}

\newcommand{\tr}{\mathrm{Tr}}

\newcommand{\s}{\op{S}}

\newcommand{\beq}{\begin{equation}}
\newcommand{\eeq}{\end{equation}}
\newcommand{\bea}[1]{\begin{equation}\begin{array}{#1}}
\newcommand{\eea}{\end{array}\end{equation}}
\newcommand{\beqn}{\begin{eqnarray}}
\newcommand{\eeqn}{\end{eqnarray}}

\renewcommand{\rho}{\varrho}

\newcommand{\ba}{\begin{eqnarray}}
\newcommand{\be}{\begin{equation}}
\newcommand{\ee}{\end{equation}}
\newcommand{\ea}{\end{eqnarray}}
\newcommand{\ban}{\begin{eqnarray*}}
\newcommand{\ean}{\end{eqnarray*}}

{\begin{framed}\begin{small}}
{\end{small}\end{framed}}

\pgfmathsetmacro{\rad}{.7}
\newcommand{\prep}[3]{
\draw [thick, black] (#1-\rad,#2) arc [radius=\rad, start angle=180, end angle= 360];
 \draw [thick] (#1-\rad,#2) -- (#1+\rad,#2);
 \node at (#1,#2-.4*\rad) {#3};}

 \newcommand{\meas}[3]{
\draw [thick, black] (#1-\rad,#2) arc [radius=\rad, start angle=180, end angle= 0];
 \draw [thick] (#1-\rad,#2) -- (#1+\rad,#2);
 \node at (#1,#2+.4*\rad) {#3};}

 \newcommand{\earth}[2]{
 \draw [thick] (#1-\rad *.65,#2+\rad *1.2-\rad *1.2) -- (#1+\rad *.65,#2+\rad *1.2-\rad *1.2);
 \draw [thick] (#1-\rad *.45,#2+\rad *1.35-\rad *1.2) -- (#1+\rad *.45,#2+\rad *1.35-\rad *1.2);
 \draw [thick] (#1-\rad *.25,#2+\rad *1.5-\rad *1.2) -- (#1+\rad *.25,#2+\rad *1.5-\rad *1.2);
 \draw [thick] (#1-\rad *.05,#2+\rad *1.65-\rad *1.2) -- (#1+\rad *.05,#2+\rad *1.65-\rad *1.2);
 }

\DeclareGraphicsExtensions{.pdf,.png,.jpg}

\makeindex

\begin{document}

\title{Information-Theoretic Implications of Quantum Causal Structures}
\date{\today}

\author{Rafael Chaves}
\affiliation{Institute for Physics, University of Freiburg, Rheinstrasse 10, D-79104 Freiburg, Germany}
\author{Christian Majenz}
\affiliation{Department of Mathematical Sciences, University of Copenhagen, Universitetsparken 5, DK-2100 Copenhagen Ø}
\author{David Gross}
\affiliation{Institute for Physics, University of Freiburg, Rheinstrasse 10, D-79104 Freiburg, Germany}

\begin{abstract}
The correlations that can be observed between a set of
variables depend on the causal structure underpinning them.
Causal structures can be modeled using directed acyclic graphs, where
nodes represent variables and edges denote functional dependencies.
In this work, we describe a general algorithm for computing
information-theoretic constraints on the correlations that can arise
from a given interaction pattern, where we allow for classical as well
as quantum variables.
We apply the general technique to two relevant cases:
First, we show that the principle of \emph{information causality}
appears naturally in our framework and go on to generalize and
strengthen it.
Second, we derive bounds on the correlations that can occur in a
networked architecture, where a set of few-body quantum systems is
distributed among a larger number of parties.
\end{abstract}

\maketitle

\section{Introduction}

A \emph{causal structure} for a set of classical variables is a graph,
where every variable is associated with a node and a directed edge
denotes functional dependence.
Such a causal model offers a means of \emph{explaining} dependencies
between variables, by specifying the process that gave rise to them.
More formally, variables $X_1, \dots, X_n$ form a
\emph{Bayesian network} with respect to a directed,
acyclic graph,
if every variable $X_i$
depends only on its graph-theoretic parents
$\operatorname{pa_i}$. This is the case
\cite{Pearlbook,Spirtesbook}
if and only if
the distribution factorizes as in
\begin{equation}\label{eqn:factorized}
	p(x_1, \dots, x_n) = \prod_{i=1}^n p (x_i | x_{\mathrm{PA}_{i}} ).
\end{equation}
One can ask the following fundamental question: \emph{Given a subset of
of variables, which correlations between them are compatible with a
given causal structure?}
In this work, we measure ``correlations'' in terms of the collection
of joint entropies of the the variables (a precise definition will be
given below).

This problem appears in several contexts. In the young field of \emph{causal
inference}, the goal is to learn causal dependencies from empirical
data
\cite{Pearlbook,Spirtesbook}.
If observed correlations are incompatible with a presumed causal
structure, it can be discarded as a possible model.
This is close to the reasoning employed in \emph{Bell's Theorem} \cite{Bell1964}  -- a
connection which is increasingly appreciated among quantum physicists \cite{Spekkens2012,Fritz2012,FritzChaves2013,Chaves2014,Fritz2014,Henson2014,Pienaar2014}.
In the context of \emph{communication theory},
these joint entropies describe the capacities that can be achieved in
network coding protocols \cite{Yeung2008}.

In this work, we are interested in quantum generalizations of causal
structures. Nodes are now allowed to represent either quantum or
classical systems, and edges are quantum operations.
An important conceptual difference to the purely classical setup
is rooted in the fact that quantum operations disturb their input. Put
differently, quantum mechanics does not assign a joint state to the
input and the output of an operation. Therefore, there is in general
no analogue to (\ref{eqn:factorized}), i.e., a global density
operator for all nodes in a quantum causal structure cannot be
defined. However, if we pick a set of nodes that do coexist (e.g.\
because they are classical, or because they are created at the same
instance of time), then we can again ask: \emph{Which joint entropies
of coexisting nodes can result from a given quantum causal structure?}
The main contribution of this work is to describe a systematic
algorithm for answering this question, generalizing previous results on
the classical case \cite{Chaves2012,FritzChaves2013,Chaves2013entropic,Chaves2014,Chaves2014b}.
We illustrate the versatility and practical relevance with two
examples. The details, along with more examples including, e.g.\ dense
coding schemes \cite{Bennett1992}, are presented in the main text.

\emph{Distributed architectures}.---Consider a scenario where, in a
first step,
several few-body quantum states are distributed among a number of
parties.
In a second step, each party processes those parts of the states it
has access to (e.g.\ by performing a coherent operation or a joint measurement).
Such setups are studied e.g.\ in distributed quantum computing
\cite{van2014quantum,buhrman2003distributed},
quantum networks
\cite{Acin2007entanglement}, quantum non-locality \cite{Brunner2014review}, and quantum
repeaters \cite{Sangouard2011}.
Which limits on the resulting correlations are implied by the network
topology alone? Our framework can be used to compute these
systematically. We will, e.g., prove certain \emph{monogamy relations}
between the correlations that can result from measurements
on distributed quantum states.

\emph{Information causality}.---The ``no-signalling principle'' alone
is insufficient to explain the ``degree of non-locality'' exhibited by
quantum mechanics \cite{Popescu1994}. This has motivated the search for stronger,
operationally motivated principles, that may single out quantum
mechanical correlations
\cite{Van1999nonlocality,Brassard2006limit,Pawlowski2009,
gross2010all,de2012deriving,Navascues2010glance,Fritz2013local,Sainz2014,Navascues2014almost}.
One of these is \emph{information causality} (IC)
\cite{Pawlowski2009,Barnum2010entropy}
which posits that an $m$ bit message from Alice to Bob must not allow
Bob to learn more than $m$ bits about a string held by Alice. A
precise formulation of the protocol involves a relatively complicated
quantum causal structure (Fig.~\ref{fig:triangle_and_IC}b). It implies
an information-theoretic bound on the mutual information between bits
$X_i$ held by Alice and guesses $Y_i$ of these by Bob
\cite{Pawlowski2009}. Here, we note that the IC setup falls into our
framework and we put the machinery to use to generalize and strengthen
it. We will show below that by taking additional information into
account, our strengthened IC principle can identify super-quantum
correlations that could not have been detected in the original
formulation.

\begin{figure} [!t]
\centering
\includegraphics[width=0.49\textwidth]{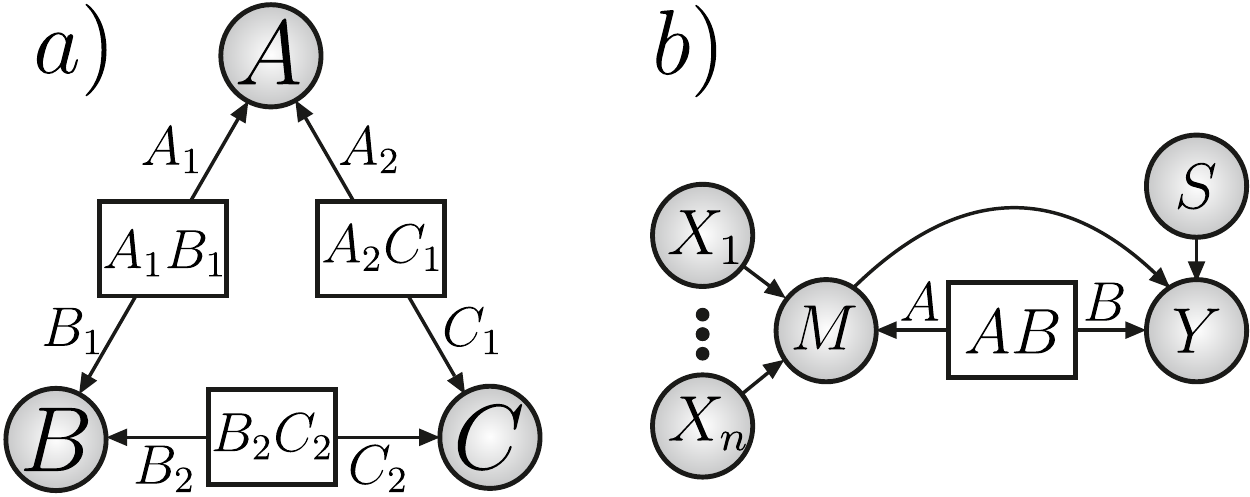}
\caption{
\textbf{(a)} An example of distributed architecture involving bipartite entangled states. Each of the underlying quantum states can connect at most two of the observable variables, what implies a non-trivial monogamy of correlations as captured in \eqref{ineq_mn}. \textbf{b} The quantum causal structure associated with the information causality principle.}
\label{fig:triangle_and_IC}
\end{figure}

\section{Quantum Causal Structures}

Informally, a quantum causal structure
specifies the functional dependency between a collection of quantum
and classical variables.
We find it helpful to employ a graphical notation, where we aim to
closely follow the conventions of classical graphical models \cite{Pearlbook,Spirtesbook}.
There are two basic building blocks: Root nodes are labeled by
a set of quantum systems and represent a density operator for these
systems
\begin{equation*}
	\raisebox{-.10cm}{\includegraphics[width=.6cm]{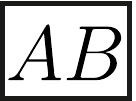}}
	\>
	\triangleq \rho_{AB}.
\end{equation*}
The second type is given by
nodes with incoming edges. Again, both the edges and
the node carry the labels of quantum systems. Such symbols represent a
quantum operation (completely positive, trace-preserving map) from the
systems associated with the edges to the ones
associated with the node:
\begin{equation*}
	\raisebox{-.23cm}{\includegraphics[width=1.0cm]{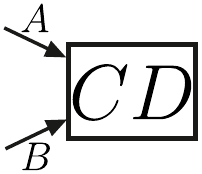}}
	\>
	\triangleq
	\Phi_{AB \to CD}: A\otimes B \to C \otimes D.
\end{equation*}
These blocks may be combined: a node containing a system $X$
can be connected to an edge with the same label. The interpretation
is, of course, that $X$ serves as the input to the associated
operation. For example,
\begin{equation*}
	\raisebox{-0.60cm}{\includegraphics[width=2.0cm]{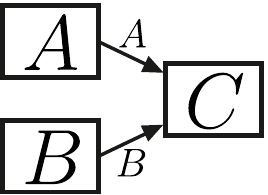}}
	\>
	\triangleq
	\rho_C = \Phi_{AB \to C} (\rho_A \otimes \rho_B)
\end{equation*}
says that the state of system $C$ is the result of applying an
operation $\Phi_{AB \to C}$ to a product state on $AB$.
To avoid ambiguities, we will never use the same label in two
different nodes (in particular, we always assume that the output
systems of an operation are distinct from the input systems).
For a more involved example,
note that Fig.~\ref{fig:triangle_and_IC}(a)
gives a fairly readable representation of the following cumbersome
algebraic statement:
\begin{eqnarray}
	\label{eqn:cumbersome}
	\rho_{ABC} &=&
	\big[
		\Phi_{A_1 A_2 \to A}
		\otimes
		\Phi_{B_1 B_2 \to B}
		\otimes
		\Phi_{C_1 C_2 \to C}
	\big] \\
	&&
	(
		\rho_{A_1 B_1} \otimes
		\rho_{A_2 C_2} \otimes
		\rho_{B_2 C_2}
	) \nonumber
\end{eqnarray}
(where the operation defined in the first line is acting on the state
defined in the second line).
The graphical representation does not indicate \emph{which} input
state or \emph{which} operation to employ. We suppress this
information, because we will be interested only in constraints on the
resulting correlations that are implied by the topology of the
interactions alone, regardless of the choice of states and maps.

We will use round edges to denote classical variables (equivalently,
quantum systems described by states which are diagonal in a given
basis).
In principle, classical variables could have more than one outgoing
edge, though this does not happen in the examples considered here. Of
course, the no-cloning principle precludes a quantum system being used
as the input to two different operations.
Only graphs that are free of cyclic dependencies can be interpreted as
specifying a causal structure. Thus, as is the case in classical
Bayesian networks, every quantum causal structure is associated with a
directed, acyclic graph (commonly abbreviated \emph{DAG}).

We note that graphical notations for quantum processes have been used
frequently before. The most popular graphical calculus is probably the
gate model of quantum computation \cite{Nielsen2010quantum}, where,
directly opposite to our conventions, operations are nodes and systems
are edges. Quantum communication scenarios are often visualized the
same way we employ here \cite{wilde2013quantum}. The recently introduced
\emph{generalized Bayesian networks} of \cite{Henson2014} are closely related
to our system. There, the authors even allow for post-quantum
resources.

We have noted in the introduction that a classical Bayesian network
not only defines the functional dependencies between random variables,
but also provides a structural formula (\ref{eqn:factorized}) for
the joint distribution of all variables in the graph. Again, such a
joint state for all systems that appear in a quantum causal structure
is not in general defined. However, other authors have considered
quantum versions of distributions that factor as in
(\ref{eqn:factorized}) and have developed graphical notations to this
end. Well-known examples include the related constructions that go by
the name of finitely correlated states,
matrix-product states, tree-tensor networks, or projected entangled pairs
states (a highly incomplete set of starting points to the literature
is given by
\cite{fannes1992finitely,perez2007matrix,shi2006classical}).
Also, certain definitions of quantum Bayesian
networks \cite{tucci1995quantum} fall into that class.

\section{Entropic description of quantum causal structures}
The entropic description of classical-quantum DAGs can be seen as a generalization of the framework for case of purely classical variables \cite{Chaves2012,FritzChaves2013,Chaves2013entropic,Chaves2014,Chaves2014b} that consists of three main steps. In the first, we describe the constraints (given in terms of linear inequalities) over the entropies of the $n$ variables describing a DAG. In the second step one needs to add to this basic set of inequalities, the causal entropic constraints as encoded in the conditional independencies implied by the DAG. In the last step, we need to eliminate from our description all terms involving variables that are not observable. The final result of this three steps program is the description of the marginal entropic constraints implied by the model under test.

We denote the set of indices of the random variables by $[n]=\{1, \dots, n\}$ and its power set (i.e., the set of subsets) by $2^{[n]}$. For every subset $S\in 2^{[n]}$ of indices, let $X_S$ be
the random vector $(X_i)_{i\in S}$ and denote by $H(S):=H(X_S)$ the associated entropy vector (for some, still unspecified entropy function $H$). Entropy is then a function $H: 2^{[n]} \to \mathbbm{R}, \qquad S \mapsto H(S)$ on the power set.

Note that as entropies must fulfill some constraints, not all entropy vectors are possible. That is, given the linear space of all set functions denoted by $R_n$ and a function $h\in R_n$ the region of vectors in $R_n$ that correspond to entropies is given by
\begin{eqnarray*}
	\left\{ h \in R_n \,|\, h(S) = H(S) \text{ for some entropy function
	} H \right\}.
\end{eqnarray*}
Clearly, this region will depend on the chosen entropy function.

For purely classical variables, $H$ is chosen to be the Shannon entropy given by $H(X_S)=-\sum_{x_s}p(x_s)\log_2 p(x_s)$. In this case an outer approximation to the associated entropy region has been studied extensively in information theory, the so called \emph{Shannon cone} $\Gamma_n$ \cite{Yeung2008}, which is the basis of the entropic approach in classical causal inference \cite{Chaves2014b}. The Shannon cone is the polyhedral closed convex cone of set functions $h$ that respect two elementary inequalities, known as polymatroidal axioms:  The first relation is the \emph{sub-modularity} (also known as strong subadditivity) condition which is equivalent to the positivity of the conditional mutual information, e.g.  $I(A:B\vert C)= H(A,C)+H(B,C)-H(A,B,C)-H(C) \geq 0$. The second inequality  -- known as \emph{monotonicity}  -- is equivalent to the positivity of the conditional entropy, e.g. $H(A \vert B) = H(A,B)- H(B) \geq 0 $.

Here lies the first difference between the classical and quantum variables, the latter being described in terms of the quantum analog of the Shannon entropy, the von Neumann entropy $H(\rho_{A,B})=-\tr \left(\rho_{A,B} \log \rho_{A,B} \right) $. While quantum variables respect sub-modularity, the von Neumann entropy fails to commit with monotonicity. Note, however, that for sets consisting of both classical and quantum variables, monotonicity may still hold. That is because the uncertainty about a classical variable $A$ cannot be negative, even if we condition on an arbitrary quantum variable $\rho$, following then that $H(A \vert \rho) \geq 0$ \cite{Barnum2010entropy}. Furthermore, for a classical variable $A$, the entropy $H(A)$ reduce to the Shannon entropy \cite{Safi2011}.

Another important difference in the quantum case is the fact that
measurements (or more generally complete positive and trace preserving
(CPTP) maps) on a quantum state will generally destroy/disturb the
state. To illustrate that consider the classical-quantum DAG in Fig.
\ref{fig:triangle_and_IC}. Consider the classical and observable
variable $A$. It can without loss of generality be considered a
deterministic function of its parents $\rho_{A_1}$ and $\rho_{A_2}$,
as any additional local parent can be absorbed in the latter. For the
variable $A$ to assume a definite outcome, a joint CPTP map is applied
to both parents $\rho_{A_1}$ and $\rho_{A_2}$ that will in general
disturb these variables. The variable $A$ does not coexist with
variables $A_1$ and $A_2$. Therefore, no entropy can be
associated to these variables simultaneously, that is,
$H(A,A_1,A_2)$ cannot be part of the entropic
description of the classical-quantum DAG. Classically, this problem
does not arise as the underlying classical hidden variables could be
accessed without disturbing them.

The elementary inequalities discussed above encode the constraints that the entropies of \emph{any} set of classical or quantum random variables are subject to. Classically, the causal relationships between the variables are encoded in the conditional independencies (CI) implied by the graph. These can be algorithmically enumerated using the so-called \emph{$d$-separation criterion} \cite{Pearlbook}. Therefore, if one further demands that classical random variables are a Bayesian network with respect to some given DAG, their entropies will also ensue the additional CI relations implied by the graph. The CIs, relations of the type $p(x,y \vert z)=p(x\vert z)p(y\vert z)$, defining non-linear constraints in terms of probabilities are faithfully translate to homogeneous linear constraints on the level of entropies, e.g. $H(X,Y \vert Z)=0$. The CIs involving jointly coexisting variables also hold for the quantum causal structures considered here \cite{Henson2014}. However, some classically valid CIs may, in the quantum case, involve non coexisting variables and therefore are not valid for quantum variables. An example of that is illustrated below for the information causality scenario.

Furthermore, because terms like $H(A,A_1,A_2)$ are not part of our description, we need, together with the CIs implied by the quantum causal structure, a rule telling us how to map the underlying quantum variables in their classical descendants, for example, how to map $H(A_1,A_2) \rightarrow H(A)$. This is achieved by the data processing (DP) inequality, another basic property that is valid both for the classical and quantum cases \cite{Nielsen2010quantum}. The DP inequality basically states that the information content of a system cannot be increased by acting locally on it. To exemplify, one DP inequality implied by the DAG in Fig. \ref{fig:triangle_and_IC} is given by $I(A:B) \leq I(A_1,A_2:B_1,B_2)$, that is, the mutual information between the classical variables cannot be larger then the information shared by their underlying quantum parents.

Finally, we are interested in situations where not all joint distributions are accessible. Most commonly, this is because the variables of a DAG can be divided into observable and not directly observable ones (e.g. the underlying quantum states in Fig. \ref{fig:triangle_and_IC}). Given the set of observable variables, in the classical case, it is natural to assume that any subset of them can be \emph{jointly} observed. However, in quantum mechanics that situation is more subtle. For example, position $Q$ and momentum $P$ of a particle are individually measurable, however, there is no way to consistently assign a joint distribution to both position and momentum of the same particle \cite{Bell1964}. That is while $H(Q)$ and $H(P)$ are part of the entropic description of classical-quantum DAGs, joint terms like $H(Q,P)$ cannot be part of it. This motivates the following definition: Given a set of variables $X_{1}, \dots, X_{n}$ contained in a DAG, a \emph{marginal scenario} $\mathcal{M}$ is the collection of those subsets of $X_1, \dots, X_n$ that are assumed to be jointly measurable.

Given the inequality description of the DAG and the marginal scenario $\mathcal{M}$ under consideration, the last step consists of eliminating from this inequality description, the variables that are not directly observable, that is the variables that are not contained in $\mathcal{M}$. This is achieved, for example, via a Fourier-Motzkin (FM) elimination (see appendix for further details). In two of the examples below (information causality and quantum networks), all the observable quantities correspond to classical variables, corresponding, for example, to the outcomes of measurements performed on quantum states. Therefore, the marginal description will be given in terms of linear inequalities involving Shannon entropies only. For the super dense coding case, the final description involves a quantum variable, therefore implying a mixed inequality with Shannon as well von Neumann entropy terms.

\section{Information Causality}
The IC principle can be understood as a kind of game: Alice receives a bit string $x$ of length $n$, while Bob receives a random number $s$ ($1 \leq s \leq n$). Bob's task is to make a guess $Y_s$ about the $s$th bit of the bit string $x$ using as resources i) a $m$-bit message $M$ sent to him by Alice and ii) some correlations shared between them. It would be expected that the amount of information available to Bob about $x$ should be bounded by the amount of information contained in the message, that is, $H(M)$. IC makes this notion precise, stating that the following inequality is valid in quantum theory \cite{Pawlowski2009}
\begin{equation}
\label{IC1}
\sum_{s=1}^{n} I(X_s:Y_s) \leq H(M)
\end{equation}
where $I(X:Y)$ is the classical mutual information between the variables $X$ and $Y$ and the input bits of Alice are assumed to be independent. This inequality is valid for quantum correlations but is violated by all nonlocal correlations beyond Tsirelson's bound \cite{Pawlowski2009,Barnum2010entropy,Dahlsten2012tsirelson}.

Consider the case where $X=(X_1,X_2)$ is a 2-bit string. The corresponding causal structure to the IC game is then the one shown in Fig. \ref{fig:triangle_and_IC} b). The only relevant CI is given by $I(X_1,X_2: AB)=0$. Note that classically the CI $I(X_1,X_2:Y_s \vert M,B)=0$ (with $s=1,2$) would also be part of our entropic description. However, because we cannot assign a joint entropy to $Y_s$ and $\rho_B$, that is not possible in quantum case anymore. We can now proceed with the general framework. But before doing that we first need to specify in which marginal scenario we are interested. In Ref. \cite{Pawlowski2009} the authors implicitly restricted their attention to the marginal scenario defined by $\left\{ X_1,Y_1 \right\},\left\{ X_2,Y_2 \right\},\left\{M \right\}$. Proceeding with this marginal scenario we find that the only non-trivial inequality characterizing this marginal entropic cone is given by
\begin{equation}
\label{IC2}
I(X_1:B_1)+I(X_2:B_2) \leq H(M)+I(X_1:X_2),
\end{equation}
that corresponds exactly to the IC inequality obtained in \cite{Safi2011} where the input bits are not assumed to be independent.

Note, however, that using the aforementioned marginal scenario, available information is being discarded. The most general possible marginal scenario is given by $\left\{ X_1,X_2,Y_s,M \right\}$ (with $s=0,1$). That is, in this case we are also interested in how much information the guess $Y_1$ of the bit $X_1$ together with the message $M$ may contain about the bit $X_2$ (similarly for $B_2$ and $X_1$). Proceeding with this marginal scenario we find different classes of non-trivial tight inequalities describing the marginal information causality cone. Of particular relevance is the following tighter version of the original IC inequality
\begin{eqnarray}
\label{ICtighter}
\nonumber
& I(X_1:Y_1,M)+I(X_2:Y_2,M) +I(X_1:X_2 \vert Y_2,M) \\
& \leq H(M)+I(X_1:X_2).
\end{eqnarray}
Two different interpretations can be given to this inequality: as a monogamy of correlations or as a classical quantification of causal influence.

For the first interpretation, consider for simplicity the case where the input bits are independent, that is, $I(X_1:X_2)=0$. These independent variables may, however, become correlated given we know the values of other variables that depend on them. That is, in general $I(X_1:X_2 \vert Y_2,M) \neq 0$. However, the underlying causal relationships between the variables impose constraints on how much we can correlate these variables. In fact, as we can see from \eqref{ICtighter}, the more information the message $M$ and the guess $Y_i$ contain about about the input bit $X_i$, the smaller is the correlation we can generate between the input bits. As an extreme example suppose Alice decides to send $M=X_1\oplus X_2$. Then $X_1$ and $X_2$ are fully correlated given $M$, but $M$ doesn't contain any information about the individual inputs $X_1$ and $X_2$.

As for the second interpretation, we need to rely on the classical concept of how to quantify causal influence between two sets of variables $X$ and $Y$. As shown in \cite{Janzing2013}, a good measure $\mathcal{C}_{X \rightarrow Y}$ of the causal influence of a variable $X$ over a variable $Y$ should be lower bounded as $\mathcal{C}_{X \rightarrow Y} \geq I(X:Y \vert Pa^X_Y)$, where $Pa^X_Y$ stands for all the parents of $Y$ but $X$. That is, excluded the correlations between $X$ and $Y$ that are mediated via $Pa^X_Y$, the remaining correlations give a lower bound to the direct causal influence between the variables. Consider for instance that we allow for an arrow between the input bits $X$ and the guess $Y$. Therefore, the classical CI $I(X_1,X_2:Y_1,Y_2 \vert M,B)=0$ that is valid for the DAG in Fig. \ref{fig:triangle_and_IC} b), does not hold any longer. In this case $I(X:Y \vert Pa^X_Y)=I(X_1,X_2:Y_1,Y_2 \vert M,B)$, an object that is part of the entropic description in the classical case. Proceeding with the general framework one can prove that
\begin{eqnarray}
\label{causal_interpretation}
\nonumber
& \mathcal{C}_{X \rightarrow Y} \geq I(X_1:Y_1,M)+I(X_2:Y_2,M)  \\
& +I(X_1:X_2 \vert Y_2,M) - H(M)- I(X_1:X_2).
\end{eqnarray}
That is, the degree of violation of \eqref{ICtighter} (for example, via a PR-box) gives exactly the minimum amount of direct causal influence required to obtain the same level of correlations within a classical model.

Inequality \eqref{ICtighter} refers to the particular case of two input bits for Alice. As we prove in the appendix the following generalization for any number of input bits is valid within quantum theory:
\begin{eqnarray}
\label{ICtighter2}
\nonumber
\sum^{n}_{i=1} I(X_i:Y_i,M) + \sum^{n}_{i=2} I(X_1:X_i \vert Y_i,M) \\
\leq H(M)+ \sum^{n}_{i=1}H(X_i)-H(X_1,\dots,X_{n}).
\end{eqnarray}

We further notice that the IC scenario is quite similar to the super dense coding scenario \cite{Bennett1992}, the only difference being on the fact that for the latter the message $M$ is a quantum state. On the level of the entropies this difference is translated in the fact that the monotonicity $H(M\vert X_0,X_1,B) \geq 0$ must be replaced by a the weak monotonicity $H(M\vert X_0,X_1,B) +H(M) \geq 0$. As proved in the appendix this implies that a similar inequality \eqref{ICtighter2} is a also valid for the super dense coding scenario if one replaces $H(M)$ by $2H(M)$, that is, a quantum message (combined with the shared entangled state) may allow for the double of information to be transmitted.

Finally, to understand how much more powerful inequality \eqref{ICtighter} may be in order to witness postquantum correlations we perform a similar analysis to the one in Ref. \cite{Allcock2009recovering}. We consider the following section of the nonsignalling polytope
\begin{equation}
\label{sec_poly}
p(a,b \vert x,y)=\gamma P_{PR}+\epsilon P_{det}+(1-\gamma-\epsilon) P_{white}
\end{equation}
with $P_{PR}(a,b \vert x,y)=(1/2)\delta_{a\oplus b, xy}$, $P_{white}(a,b \vert x,y)=1/4$ and $P_{det}(a,b \vert x,y)=\delta_{a,0}\delta_{b,0}$ corresponding, respectively to the PR-box, white noise and a deterministic box. The results are shown in Fig. \ref{fig:ICplot} where it can be seen that the new inequality is considerably more powerful then the original one. It can for instance witness the postquantumness of distributions that could not be detected before even in the limit of many copies.

\begin{figure} [!t]
\centering
\includegraphics[width=0.45\textwidth]{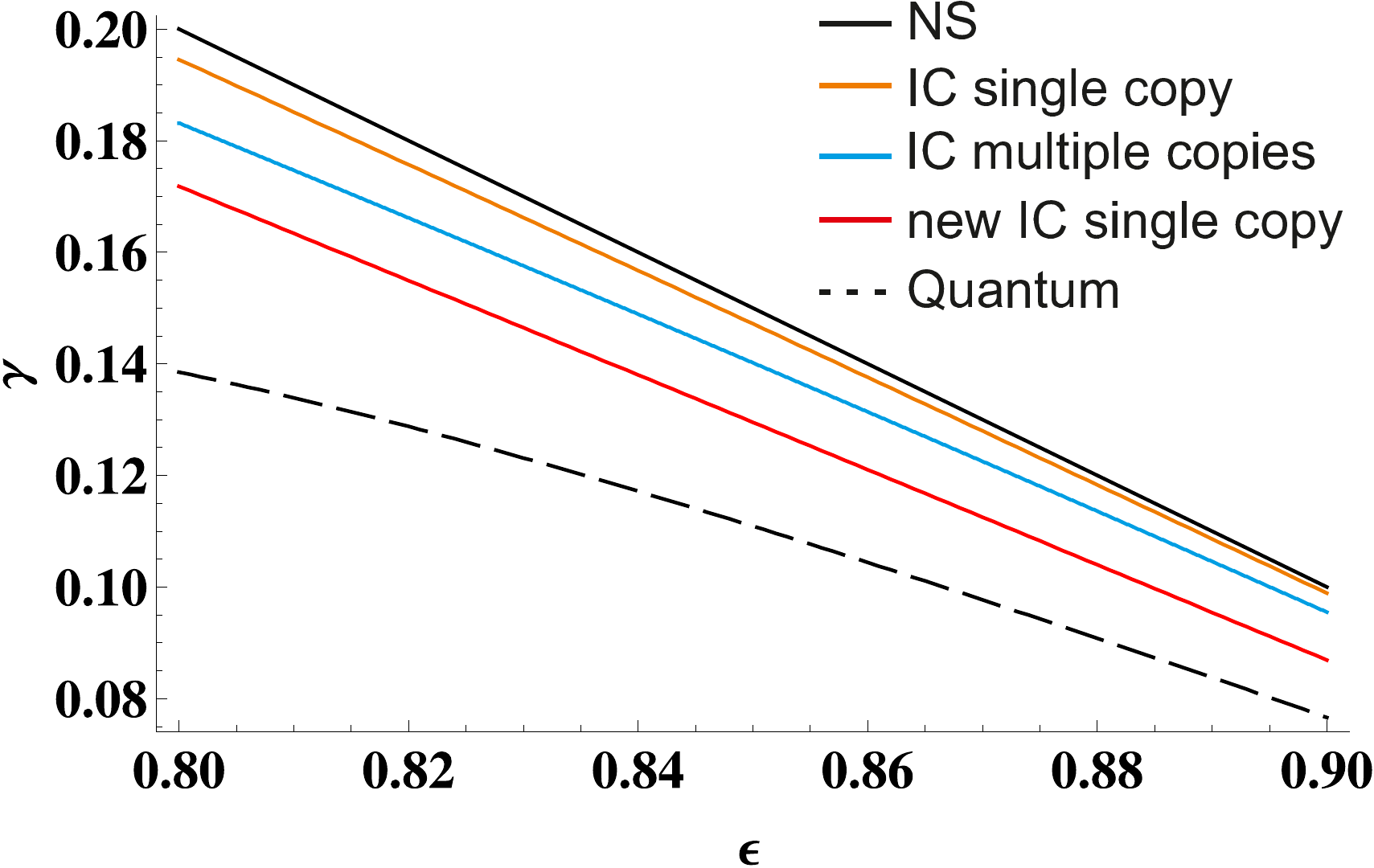}
\caption{A slice of the non-signalling polytope corresponding to the distributions \eqref{sec_poly}. The lower black dashed line is an upper limit on quantum correlations obtained via the criterion in Ref. \cite{Navascues2007} while the upper solid black line bounds the set of non-signalling correlations. The solid red, blue and orange curves correspond, respectively, to the boundaries obtained with the IC inequalities \eqref{ICtighter}, \eqref{IC1} and \eqref{IC2}. Above each of this curves, the corresponding inequalities are violated. See appendix for details of how this curves are computed.
}
\label{fig:ICplot}
\end{figure}

\section{Quantum networks} Quantum networks are ubiquitous in quantum
information. The basic scenario consists of a collection of entangled
states that are distributed among several spatially separated parties
in order to perform some informational task, e.g., entanglement
percolation \cite{Acin2007entanglement}, entanglement swapping
\cite{Zukowski1993} or distributed computing
\cite{van2014quantum,buhrman2003distributed}.
A similar setup
is of relevance in classical causal inference, namely the inference of
latent common ancestors \cite{Steudel2010,Chaves2014b}. As we will
show next, the topology alone of these quantum networks imply
non-trivial constraints on the correlations that can be obtained
between the different parties. We will consider the particular case
where all the parties can be connected by at most bipartite states. We
note, however, that our framework applies as well to the most general
case and results along this line are presented in the appendix.

The problem can be restated as follows. Consider $n$ observable variables that may be assumed to have no direct causal influence on each other (as they are space-like separated). Given some observed correlations between them, the basic question is then: Can the correlations between these $n$ variables be explained by (hidden) common ancestors connecting at most $2$ of them? The simplest of such common ancestors scenarios ($n=3$), the so called triangle scenario \cite{Steudel2010,Branciard2012,Fritz2012}, is illustrated in Fig. \ref{fig:triangle_and_IC}.

In the case where the underlying hidden variables are classical (for example, separable states), the entropic marginal cone associated to this DAG has been completely characterized in Ref. \cite{Chaves2014}. Following the framework delineated before, we can prove that the same cone is obtained if we replace the underlying classical variables by quantum states (see appendix). This implies that \emph{entropically there are no quantum correlations in the triangle scenario}.

The natural question is how to generalize this result to more general common ancestor structures for arbitrary $n$. With this aim, we prove in the appendix that the monogamy relation
\begin{equation}
\label{ineq_mn}
\sum_{\substack{
   i=1,\cdots,n \\
   i \neq j
  }}I(V_{i}:V_{j})\leq H(V_{j}),
\end{equation}
recently derived in \cite{Chaves2014b} is also valid for quantum theory. We also prove in the appendix that this inequality is valid for general non-signalling theories, generalizing the result obtained in \cite{Henson2014} for $n=3$. In addition we exhibit that for any nontrivial common ancestor structure there are entropic corollaries even if we allow for general non-signalling parents.

The inequality \eqref{ineq_mn} can be seen as a kind of monogamy of correlations. Consider for instance the case $n=3$ and label the commons ancestor (any nonsignalling resource) connecting variables $V_i$ and $V_j$ by $\rho_{i,j}$. If the dependency between $V_1$ and $V_2$ is large, that means that $V_1$ has a strong causal dependence on their common mutual ancestor $\rho_{1,2}$. That implies that $V_1$ should depend only mildly on its common ancestor $\rho_{1,3}$ and therefore its correlation with $V_3$ should also be small. The inequality \eqref{ineq_mn} makes this intuition precise.

\section{Discussion}
\label{sec:discussion}

In this work, we have introduced a systematic algorithm for computing information-theoretic constraints arising from quantum causal structures. Moreover, we have demonstrated the versatility of the framework by applying it to a set of diverse examples from quantum foundations, quantum communication, and the analysis of distributed architectures. In particular, our framework readily allows to obtain a much stronger version of information causality.

These examples aside, we believe that the main contribution of this work is to highlight the power of systematically analyzing entropic marginals. A number of future directions for research immediately suggest themselves. In particular, it will likely be fruitful to consider multi-partite versions of information causality or other information theoretical principles and to further look into the operational meaning of entropy inequality violations.

\bigskip
\begin{acknowledgements}
We acknowledge support by the Excellence Initiative of the German
Federal and State Governments (Grant ZUK 43), the Research Innovation
Fund from the University of Freiburg. DG's research is supported by
the US Army Research Office under contracts W911NF-14-1-0098 and
W911NF-14-1-0133 (Quantum Characterization, Verification, and
Validation). CM acknowledges support by the German National Academic
Foundation.
\end{acknowledgements}

\section{Appendix}

\subsection{A linear program framework to entropic inequalities}

Given the inequality description of the entropic cone describing a causal structure, to obtain the description of an associated marginal scenario $\mathcal{M}$ we need to eliminate from the set of inequalities all variables not contained in $\mathcal{M}$. After this elimination procedure, we obtain a new set of linear inequalities, constraints that correspond to facets of a convex cone, more precisely the marginal entropic cone characterizing the compatibility region of a certain causal structure \cite{Chaves2014}. This can be achieved via a Fourier-Motzkin (FM) elimination, a standard linear programming algorithm for eliminating variables from systems of inequalities \cite{Williams1986}. The problem with the FM elimination is that it is a double exponential algorithm in the number of variables to be eliminated. As the number of variables in the causal structure of interest increases, typically this elimination becomes computationally intractable.

While it can be computationally very demanding to obtain the full description of a marginal cone, to check if a given candidate inequality is respected by a causal structure is relatively easy. Consider that a given causal structure leads to a number $N$ of possible entropies. These are organized in a $n$-dimensional vector $\boldsymbol{h}$. In the purely classical case, the graph consisting of $n$ nodes ($X_1,\dots, X_n$) will lead to a $N=2^n$ dimensional entropy vector that can be organized as $\boldsymbol{h}=(H(\emptyset), H(X_n), H(X_{n-1}),H(X_{n-1}X_n),\dots,H(X_1,\dots,X_{n}))$. In the quantum case, since not all subsets of variables may jointly coexist we will have typically that $N$ is strictly smaller than $2^n$.

As explained in details in the main text, for this entropy vector to be compatible with a given causal structure, a set of linear constraints must be fulfilled. These linear constraints can be casted as a system of inequalities of the form $M\boldsymbol{h} \geq \boldsymbol{0}$, where $M$ is a $m \times N$ matrix with $m$ being the number of inequalities characterizing the causal structure.

Given the entropy vector $\boldsymbol{h}$, any entropic linear inequality can be written simply as the inner product $\langle \boldsymbol{\mathcal{I}}, \boldsymbol{h} \rangle \geq 0$, where $\boldsymbol{\mathcal{I}}$ is the associated vector to the inequality. A sufficient condition for a given inequality to be valid for a given causal structure is that the associated set of inequalities $M\boldsymbol{h} \geq \boldsymbol{0}$ to be true for any entropy vector $\boldsymbol{h}$. That is, to check the validity of a test inequality, one simply needs to solve the following linear program:
\begin{eqnarray}
\underset{\boldsymbol{h} \in \RR^N}{\minimize} & & \langle \boldsymbol{\mathcal{I}}, \boldsymbol{h} \rangle  	\label{LP} \\
\st & &  M\boldsymbol{h} \geq \boldsymbol{0} \nonumber
\end{eqnarray}

In general, this linear program only provides a sufficient but not necessary condition for the validity of a inequality. The reason for that is the existence of non-Shannon type inequalities, that are briefly discussed below.

\subsection{Details about the new IC inequality}
In the following we will discuss how to characterize the most general marginal scenario in the information causality scenario. We will start discussing the purely classical case (i.e Alice and Bob share classical correlations) and afterwards apply the linear program framework to prove that all inequalities characterizing the classical Shannon cone are also valid for quantum mechanical correlations.

The classical causal structure associated with information causality contains six classical variables $S=\left\{X_1,X_2,Y_1,Y_2,M,\lambda \right\}$. The variable $\lambda$ stands here for the classical analog of the quantum state $\rho_{AB}$. The most general marginal scenario that is compatible with the information causality game and thus with protocols using more general resources such as nonlocal boxes is given by $\mathcal{M}=\left\{ X_1,X_2,Y_i,M \right\}$ (with $i=1,2$). The relevant conditional independencies implied by the graph are given by $I(X_1,X_2:\lambda)=0$ and $I(X_1,X_2:Y_1,Y_2 \vert M, \lambda)=0$. CIs like $I(X_1:Y_1 \vert M, \lambda)=0$ are implied by the relevant ones together with the polymatroidal axioms for the set $S$ of variables, and in this sense are thus redundant. Given this inequality description (basic inequalities plus CIs) we need to eliminate from our description, via a FM elimination, all the variable not contained in $\mathcal{M}$.

Our first step was to eliminate from the system of inequalities the variable $\lambda$. Doing that one obtains a new set of inequalities for the five variables $S=\left\{X_1,X_2,Y_1,Y_2,M \right\}$. These set of inequalities is simply given by the basic inequalities plus one single non-trivial inequality, implied by the CIs:
\begin{equation}
H(Y_1,Y_2,M)+H(X_1,X_2) \leq H(M) +H(X_1,X_2,Y_1,Y_2,M)
\end{equation}

We then proceed eliminating all variables not contained in $\mathcal{M}$. The final inequality description of the marginal cone of $\mathcal{M}$ can be organized in two groups. The first group contains all inequalities that are valid for the collection of variables in $\mathcal{M}$ independently of the underlying causal relationships between them, that is, they follow from the basic inequalities alone. The second group contains the inequalities that follow from the basic inequalities plus the conditional independencies implied by the causal structure. These are the inequalities capturing the causal relations implied by information causality and there are $54$ of them. Among these $54$ inequalities, one of particular relevance is the tighter IC inequality \eqref{ICtighter} given in the main text.

One can prove, using the linear program framework delineated before, that this inequality is also valid for the corresponding quantum causal structure shown in Fig. \ref{fig:triangle_and_IC} b). Following the discussion in the main text, the sets of jointly existing variables in the quantum case are given by $S_0=\left\{X_1,X_2,A,B \right\}$, $S_1=\left\{X_1,X_2,M,B \right\}$ and $S_2=\left\{X_1,X_2,M,Y_i \right\}$ (with $i=1,2$). One can think about these sets of variables in a time ordered manner. At time $t=0$ the joint existing variables are the inputs $X_1$ and $X_2$ of Alice, together with the shared quantum state $\rho_{A,B}$. At time $t=1$ Alice encodes the input bits into the message $M$ also using her correlations with Bob obtained through the shared quantum state. Doing that, Alice disturbs her part $A$ of the quantum system that therefore does not coexist anymore with the variables defined in $S_1$. In the final step of the protocol at time $t=2$, Bob uses the received message $M$ and its part $B$ of the quantum state in order to make a guess $Y_1$ or $Y_2$ about Alice's inputs. Once more, by doing that $B$ ceases to coexist with the variables contained in $S_2$.

Following the general idea, we write down all the basic inequalities for the sets $S_0$ and $S_1$ and $S_2$, together with the conditional independencies and the data processing inequalities. As discussed before, because the quantum analogous of $I(X_1,X_2:Y_1,Y_2 \vert M, \lambda)=0$ has no description in the quantum case, the only CI implied here will be $I(X_1,X_2 : A,B)=0$. The causal relations encoded in the other CIs are taken care by the data processing inequalities. Below we list all used data processing inequalities:

\begin{eqnarray*}
I(X_1,X_2:Y_i)  & &  \leq I(X_1,X_2:M,B)  \\ \nonumber
I(X_i:Y_j)  & &  \leq I(X_i:M,B) \\ \nonumber
I(X_i: X_{i \oplus 1}, Y_j)  & &  \leq I(X_i: X_{i \oplus 1}, B) \\ \nonumber
I(X_1,X_2:Y_i,M)  & &  \leq I(X_1,X_2:B,M) \\ \nonumber
I(X_i:Y_j,M)  & &  \leq I(X_i:B,M) \\ \nonumber
I(X_i:X_{i \oplus 1}, Y_j, M) \leq  & &  I(X_i:X_{i \oplus 1}, B, M)
\end{eqnarray*}

Note that some of these DP inequalities may be redundant, that is, they may be implied by other DP inequalities together with the basic inequalities.

We organize all the above constraints into a matrix $M$ and given a certain candidate inequality $\mathcal{I}$ we run the linear program discussed before. Doing that one can easily prove that inequality \eqref{ICtighter} is also valid in the quantum case.

Note that this computational analysis will in general be restricted by the number of variables involved in the causal structure. To circumvent that we provide in the following an analytical proof of the validity of the generalized IC inequality \eqref{ICtighter2} for the quantum causal structure in Fig. \ref{fig:triangle_and_IC} b).
\begin{proof}
First rewrite the following conditional mutual information as
\begin{equation}
I(X_1:X_i \vert Y_i,M) = I(X_1:X_i, Y_i,M)-I(X_i: Y_i,M).
\end{equation}
The LHS of the inequality \eqref{ICtighter2} can then be rewritten as
\begin{equation}
I(X_1:Y_1,M)+\sum^{n}_{i=2}I(X_i:X_1,Y_i,M).
\end{equation}
This quantity can be upper bounded as
\begin{widetext}
\begin{eqnarray}
& \leq I(X_1:B,M) +\sum^{n}_{i=2}I(X_i:X_1,B,M) \\
& = \sum^{n}_{i=1} H(X_i)+H(B,M)+(n-2)H(X_1,B,M) - \sum^{n}_{i=2} H(X_1,X_i,B,M) \\
& \leq \sum^{n}_{i=1} H(X_i)+H(B,M) - H(X_1,\dots,X_{n},B,M) \\
& \leq \sum^{n}_{i=1} H(X_i)+H(B,M) - H(X_1,\dots,X_{n},B) \\
& = \sum^{n}_{i=1} H(X_i)+H(B,M) - H(X_1,\dots,X_{n})-H(B) \\
& \leq \sum^{n}_{i=1} H(X_i)+H(M) - H(X_1,\dots,X_{n})
\end{eqnarray}
\end{widetext}
leading exactly to the inequality \eqref{ICtighter2}. In the proof above we have used consecutively i) the data processing inequalities $I(X_1:Y_1,M) \leq I(X_1:B,M) $ and $I(X_i:X_1,Y_i,M) \leq I(X_i:X_1,B,M)$, ii) the fact that $-\sum^{N-1}_{i=1} H(X_i,X_1,B,M) \leq - H(X_1,\dots,X_{n},B,M)-(n-2)H(X_1,B,M) $ (as can be easily be proved inductively using the strong subadditivity property of entropies), iii) the monotonicity $H(M\vert X_1,\dots,X_{n},B) \geq 0$, iv) the independence relation $I(X_1,\dots, X_{n}:B)=0$ and v) the positivity of the mutual information $I(B:M) \geq 0$
\end{proof}

Note that this proof can be easily adapted to the case where the message $M$ sent from Alice to Bob is a quantum state. In this case there two differences. First, because the message is disturbed in order to create the guess $Y_i$, we cannot assign a entropy to $M$ and $Y_i$ simultaneously. That is, in the LHS side of the inequality \eqref{ICtighter2} we replace  $I(X_i:Y_i,M) \rightarrow I(X_i:Y_i) $ and $I(X_1:X_i \vert Y_i,M) \rightarrow I(X_1:X_i \vert Y_i) $. The second difference is in step iii), because we have used the monotonicity $H(M\vert X_1,\dots,X_{n},B) \geq 0$ that is not valid for a quantum message. Instead of that, we can use a weak monotonicity inequality, namely $H(M\vert X_1,\dots,X_{n},B) + H(M) \geq 0$. Therefore, in the final inequality \eqref{ICtighter2}, $I(X_i:Y_i,M) \rightarrow I(X_i:Y_i) $ and $I(X_1:X_i \vert Y_i,M) \rightarrow I(X_1:X_i \vert Y_i)$ and $H(M)$ is replaced by $2H(M)$ (here, $H$ standing for the von Neumann entropy), leading exactly to what should be expected of a super dense coding \cite{Bennett1992}.

\subsection{Proving that in the triangle scenario the classical and quantum marginal cones coincide}
Since 1998 it is known that, for a number of variables $n \geq 4$, there are inequalities valid for Shannon entropies that cannot be derived from the elemental set of polymatroidal axioms (submodularity and monotonicity) \cite{Zhang1998}. These are the so called non-Shannon type inequalities \cite{Yeung2008}. More precisely, the existence of these inequalities imply that the true entropic cone (denoted by $\overline{\Gamma}^{*}_n$) is a strict subset of the Shannon cone, that is, the inclusion $ \overline{\Gamma}^{*}_n \subseteq \Gamma_n$ is strict for $n \geq 4$.

Remember that a convex cone has a dual description, either in terms of its facets or its extremal rays. In terms of its half-space description the strict inclusion $ \overline{\Gamma}^{*}_n \subset \Gamma_n$ implies that while all Shannon type inequalities are valid for any true entropy vector, they may fail to be tight. In terms of the extremal rays, this implies that some of the extremal rays of the Shannon cone are not populated, that is, there is no well defined probability distribution with an entropy vector corresponding to it.

Sometimes, the projection of the outer approximation $\Gamma_n$ onto a subspace, described by the marginal cone $\Gamma_{\mathcal{M}}$, may lead to the true cone in the marginal space, that is $\Gamma_{\mathcal{M}}=\overline{\Gamma}^{*}_{\mathcal{M}}$ \cite{FritzChaves2013}. A sufficient condition for that to happen is that all the extremal rays of $\Gamma_{\mathcal{M}}$ are populated. Using this idea, in the following we will prove that all the extremal rays describing the Shannon marginal cone of classical triangle scenario are populated, proving that in this case the Shannon and true marginal cones coincide. We will then use the linear program framework delineated previously in order to prove that all the corresponding inequalities are also valid for underlying quantum states, therefore proving that entropically the set of classical and quantum correlations coincide.

Proceeding with the three steps program delineated in the main text, one can see that the marginal scenario $\left\{ A, B, C \right\}$ of the triangle scenario is completely characterized by the following non-trivial Shannon type inequalities (and permutations thereof) \cite{Chaves2014}
\begin{widetext}
\begin{eqnarray}
\label{triangle_nontrivial_1}
&& I(A:B)+I(A:C) -H(A) \leq 0,  \\
\label{triangle_nontrivial_2}
&& I(A:B:C)+I(A:B)+I(A:C)+I(B:C) -H(A,B) \leq 0, \\
\label{triangle_nontrivial_3}
&& I(A:B:C)+I(A:B)+I(A:C)+I(B:C)- \frac{1}{2}( H(A)+H(B)+H(C) ) \leq 0,
\end{eqnarray}
\end{widetext}
plus the polymatroidal axioms for the three variables $A,B,C$. Given the inequality description of the marginal cone we have used the software PORTA \cite{porta} in order to recover the extremal rays of it. There are only $10$ extremal rays, that can be organized in the $4$ different types listed in Table \ref{extremal_rays} .

\begin{table*}
\begin{tabular}{|c| c| c c c c c c c|} \hline
\multicolumn{9}{|c|}{Extremal rays of the marginal triangle scenario}\\
\hline
type
&\textbf{\#} of permutations
&$H_{C}$&$H_{B}$&$H_{BC}$&$H_{A}$&$H_{AC}$&$H_{AB}$&$H_{ABC}$
\\
\hline
1
&3
&0&0&0&1&1&1&1
\\
\hline
2
&3
&0&1&1&1&1&1&1
\\
\hline
3
&1
&1&1&2&1&2&2&2
\\
\hline
4
&3
&3&3&5&2&4&4&6
\\
\hline
\end{tabular}
\caption{The four kinds of extremal rays defining the marginal entropic cone of the triangle scenario.} \label{extremal_rays}
\end{table*}
Below we list the probability distributions reproducing the $4$ different types of entropy vectors:
\begin{widetext}
\begin{equation}
\label{type1}
p_1\left(  a,b,c \right)  =\left\{
\begin{array}{ll}
1/2 & \text{if } a=\left\{1,2\right\} \text{ and } b=c=\left\{1\right\}\\
0 & \text{, otherwise}%
\end{array}
\right. ,
\end{equation}

\begin{equation}
\label{type2}
p_2\left(  a,b,c \right)  =\left\{
\begin{array}{ll}
1/2 & \text{if } a=b=\left\{1,2\right\} \text{ and } c=\left\{1\right\}\\
0 & \text{, otherwise}%
\end{array}
\right. ,
\end{equation}

\begin{equation}
\label{type3}
p_3\left(  a,b,c \right)  =\left\{
\begin{array}{ll}
1/4 & \text{if } a\oplus b \oplus c=0 \text{ with } a,b,c= \left\{ 1,2 \right\}\\
0 & \text{, otherwise}%
\end{array}
\right. ,
\end{equation}

and

\begin{equation}
\label{type4}
p_4\left(  a,b,c \right)  =\left\{
\begin{array}{ll}
1/64 & \text{if } a\oplus b + a\oplus c + b\oplus c=0 \text{ with } a=\left\{1,\dots,4\right\} \text{ and } b,c=\left\{1,\dots,8\right\}\\
0 & \text{, otherwise}%
\end{array}
\right. ,
\end{equation}
\end{widetext}

Since all the extremal rays are populated, this proves that $\Gamma_{\mathcal{M}}=\overline{\Gamma}^{*}_{\mathcal{M}}$ for the marginal scenario $\mathcal{M}= \left\{A,B,C \right\}$ of the triangle scenario.

To prove that the same entropic cone holds for the associated quantum causal structure (Fig. \ref{fig:triangle_and_IC} a) we just need to prove that all the inequalities defining $\Gamma_{\mathcal{M}}$ hold true in the quantum case. Clearly, the polymatroidal axioms for $\left\{A,B,C \right\}$ also hold true in the quantum case, since these variables are classical. Using the linear programming framework detailed above, one can also prove that the inequalities \eqref{triangle_nontrivial_1}, \eqref{triangle_nontrivial_2} and \eqref{triangle_nontrivial_3} hold if the underlying hidden variables stand for quantum states.

The sets of jointly existing variables in the quantum case are given by $S_0=\left\{A_1,A_2, B_1,B_2 , C_1,C_2 \right\}$, $S_1=\left\{A, B_1,B_2 , C_1,C_2 \right\}$, $S_2=\left\{A_1,A_2, B, C_1,C_2 \right\}$, $S_3=\left\{A_1,A_2, B_1,B_2 , C \right\}$, $S_4=\left\{A, B , C_1,C_2 \right\}$, $S_5=\left\{A, B_1,B_2 , C \right\}$, $S_6=\left\{A_1,A_2, B, C \right\}$ and $S_7=\left\{A, B, C \right\}$. The fact that the quantum states are assumed to be independent is translated in the CI $H(A_1,A_2, B_1,B_2 , C_1,C_2 )= H(\rho_{A_1,B_1})+ H(\rho_{A_2,C_1})+ H(\rho_{C_2,C_2})$. The causal constraint that the observable variables have no direct influence on each other (all the correlation are mediated by the underlying quantum states) is encoded in the CI given by $I(A:B \vert A_1)=I(A:B \vert B_1)=0$ (similarly for permutation of the variables). Below we also list all used data processing inequalities:
\begin{eqnarray*}
I(A:B)  & &  \leq I(A:B_1 ,B_2)  \\ \nonumber
I(A:B)  & &  \leq I(A_1 ,A_2:B) \\ \nonumber
I(A:B)  & &  \leq I(A_1 ,A_2:B_1 ,B_2) \\ \nonumber
I(A:B,C)  & &  \leq I(A:B, C_1 ,C_2)  \\ \nonumber
I(A:B,C)  & &  \leq I(A:B_1 ,B_2,C) \\ \nonumber
I(A:B,C)  & &  \leq I(A:B_1 ,B_2, C_1 ,C_2)  \\ \nonumber
I(A:B,C)  & &  \leq I(A_1 ,A_2:B, C_1 ,C_2)  \\ \nonumber
I(A:B,C)  & &  \leq I(A_1 ,A_2:B_1 ,B_2,C) \\ \nonumber
I(A:B,C)  & &  \leq I(A_1 ,A_2:B_1 ,B_2, C_1 ,C_2)  \\ \nonumber
I(A,B_1:B_2)  & &  \leq I(A_1 ,A_2,B_1:B_2) \\ \nonumber
I(A,C_1:C_2)  & &  \leq I(A_1 ,A_2,C_1:C_2)
\end{eqnarray*}
and similarly for permutations of all variables. Again, note that some of these DP inequalities may be redundant, that is, they may be implied by other DP inequalities together with the basic inequalities.

\subsection{Proving the monogamy relations of quantum networks}

In the following we provide an analytical proof of the monogamy inequality \eqref{ineq_mn} in the main text.

We start with the case $n=3$. For a Hilbert space $\mathcal{H}$ we denote the set of quantum states, i.e. the set of positive semidefinite operators with trace one, on it by $\mathbb{S}(\mathcal{H})$.

\begin{thm}\label{qtri}
Let $\rho_{A_1A_2B_1B_2C_1C_2}=\rho_{A_1 B_2}\otimes\rho_{B_1 C_2}\otimes\rho_{C_1 A_2}$ be a sixpartite quantum state on $\mathcal{H}=\mathcal{H}_{A_1}\otimes\mathcal{H}_{A_2}\otimes\mathcal{H}_{B_1}\otimes\mathcal{H}_{B_2}\otimes\mathcal{H}_{C_1}\otimes\mathcal{H}_{C_2}$. Let further $\Phi_N: \mathbb{S}\left(\mathcal{H}_{N_1}\otimes\mathcal{H}_{N_2}\right)\to \mathbb{S}\left(\mathcal{H}_N\right)$ be an arbitrary measurement for $N=A,B,C$. Then
\begin{equation}
  I(A:B)+I(A:C)\leq H(A).
 \end{equation}
\end{thm}

\begin{proof}
Data processing yields
 \begin{eqnarray}
 I(A:B)+I(A:C)&\leq& I(A:B_1B_2)+I(A:C_1C_2).
 \end{eqnarray}
 Then we exploit the chain rule twice and afterwards data processing again,
\begin{eqnarray}
 I(A:B_1B_2)&=&I(A:B_2)+I(A:B_1|B_2)\nonumber\\
 &=&I(A:B_2)+I(AB_2:B_1)-I(B_2:B_1)\nonumber\\
 &\leq &I(A:B_2)+I(A_1A_2B_2:B_1)\nonumber\\
 &=&I(A:B_2).
\end{eqnarray}
We have therefore
\begin{equation}
 I(A:B)+I(A:C)\leq I(A:B_2)+I(A:C_1),
\end{equation}
from which it follows that
\begin{eqnarray}\label{eq:proof-qtri}
  &I(A:B_2)+I(A:C_1)\nonumber\\
  & = 2H(A)+H(B_2)+H(C_1)-H(AB_2)-H(AC_1)\nonumber\\
  & \leq H(A)+H(B_2)+H(C_1)-H(AB_2C_1)\nonumber\\
  & = H(A)-H(A|B_2C_1)\nonumber\\
  &\leq H(A),
\end{eqnarray}
where in the second line we used strong subadditivity and in the last line we used that the entropy of a classical state conditioned on a quantum state is positive.

\end{proof}
This proof can easily be generalized to the case of an arbitrary number of random variables resulting from a classical-quantum Bayesian network in with each parent connects has at most two children.
\begin{cor}
 Let
 \begin{equation*}
  \rho=\bigotimes_{\substack{i,j=1\\ i<j}}^n\rho_{(ij),(ji)}
 \end{equation*}
be an $n(n-1)$-partite quantum state on
\begin{equation*}
 \mathcal{H}=\bigotimes_{\substack{i,j=1\\i\neq j}}^n \mathcal{H}_{(ij)},
\end{equation*}
and let
\begin{equation*}
 \Phi_{i}: \mathbb{S}\left(\bigotimes_{j\neq i} \mathcal{H}_{(ij)}\right)\to \mathbb{S}\left(\mathcal{H}_i\right)
\end{equation*}
be an arbitrary measurement for $i=1,...,n$. Then
\begin{equation}
 \sum_{i=2}^n I(1:i)\le H(1)
\end{equation}

\end{cor}
\begin{proof}
First, utilize the independences in the same way as in the proof of Theorem \ref{qtri} to conclude
\begin{equation}
 \sum_{i=2}^n I(1:i)\le\sum_{i=2}^n I(1:(i1)).
\end{equation}
Now continue by induction. For $n=3$ we have, according to the proof of Theorem \ref{qtri},
\begin{equation}
 I(1:(21))+I(1:(31))\le H(1).
\end{equation}
Now assume
\begin{equation}
 \sum_{i=2}^{n-1} I(1:(i1))\le H(1).
\end{equation}
Using the proof of Theorem \ref{qtri} again and stopping before the last inequality in \eqref{eq:proof-qtri} we get
\begin{equation}
 I(1:(n-1 1))+I(1:(n1))\le I(1:(n-1 1)(n1)),
\end{equation}
i.e. we get
\begin{eqnarray}
  \sum_{i=2}^{n} I(1:(i1))\le \sum_{i=2}^{n-1}I(1:(i1)')\le H(1),
\end{eqnarray}
where we defined the primed systems by $\mathcal{H}_{(n-11)'}=\mathcal{H}_{(n-11)}\otimes \mathcal{H}_{(n1)}$, observing that this yields a classical-quantum bayesian network with $n-1$ nodes and and connectivity two and used the induction hypothesis.
\end{proof}

\subsection{Proving the monogamy relation for GPTs}

\begin{figure}
\begin{center}
\begin{tikzpicture}
\prep{-3}{0}{$\sigma_{AB}$}
\prep{0}{0}{$\sigma_{AC}$}
\prep{3}{0}{$\sigma_{BC}$}
\meas{-3}{2}{$X$}
 \meas{0}{2}{$Y$}
  \meas{3}{2}{$Z$}

  \draw (-3.33,0) -- (-3.33,2);
  \draw (-2.67,0) -- (-0.33,2);
  \draw (-0.33,0) -- (-2.67,2);
   \draw (3.33,0) -- (3.33,2);
  \draw (2.67,0) -- (0.33,2);
  \draw (0.33,0) -- (2.67,2);

\end{tikzpicture}
\end{center}
\caption{The GBN for the triangle}\label{gtri}
\end{figure}
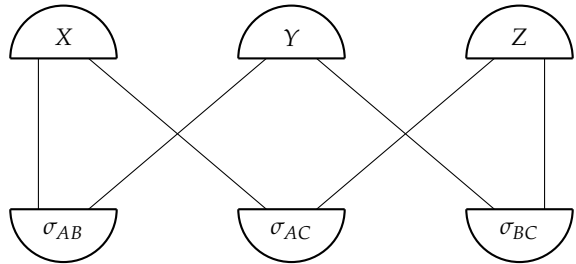

We want to prove the inequality
\begin{equation}\label{theineq}
 \sum_{\substack{j\in\{1,...,n\}\\ j\neq i}} I(V_i:V_j)\le H(V_i)
\end{equation}

For random variables that constitute a \emph{generalized Bayesian network} \cite{Henson2014} with respect to a DAG where each parent correlates at most two of them, i.e. the random variables are results of measurements on a set of arbitrary non-signalling resources shared between two parties. The case of three random variables has been proven in \cite{Henson2014}, the purpose of this appendix is to prove the generalization to an arbitrary number of random variables. Also we want to proof that for any fixed connectivity number for the parent nodes there are entropic corollaries. To this end we have to introduce a framework to handle generalized probabilistic theories that are non-signaling and have a property called local discriminability that was developed in \cite{Chiribella2010}.

An \emph{operational probabilistic theory} has two basic notions, \emph{systems} and \emph{tests}. Tests are the objects that represent any physical operation that is performed, e.g. the preparation of a state, or a measurement. A test has input and output systems and can have a  classical random variable as measurement outcome as well. An outcome together with the corresponding output system state is called \emph{event}. The components are graphically represented by a directed acyclic graph (DAG) where the nodes represent tests, and the edges represent systems. We use the convention that the diagram is read from bottom to top, i.e. a tests input systems are represented by edges coming from below and its output systems are edges emerging from the top of the node:

\begin{center}

\begin{tikzpicture}
\draw [black] (-.6,2.6) rectangle (.6,1.4);
\draw (0,3.6) -- (0,2.6);
\draw (0,1.4) -- (0,0.4);
\node at (0,2) {$X$};
\node at (.2,.9) {$A$};
\node at (.2,3.1) {$B$};
\end{tikzpicture}
\end{center}

If a system has trivial input or trivial output we omit the edge and represent the node by a half moon shape. Tests with trivial input are called \emph{preparations}, tests with trivial output are called \emph{measurements}.

\begin{center}

\begin{tikzpicture}
\draw [thick, black] (-2.5,1.2) arc [radius=1, start angle=180, end angle= 360];
\draw [thick] (-2.5,1.2) -- (-.5,1.2);
\draw [thick] (2.5,1.2) -- (.5,1.2);
\draw [thick] (.5,1.2) arc [radius=1, start angle=180, end angle=0];
\draw (1.5,0) -- (1.5,1.2);
\draw (-1.5, 1.2) -- (-1.5,2.4);
\node at (-1.5,.8) {$X$};
\node at (1.5,1.6) {$Y$};
\node at (1.7,.5) {$B$};
\node at (-1.3,1.9) {$A$};
\end{tikzpicture}
\end{center}

If we do not need to talk about the systems we omit the labels of the edges, and preparation tests are given greek letter labels, as they have, without loss of generality, only a single event.

\begin{center}
\begin{tikzpicture}
\draw [thick, black] (-2.5,1.2) arc [radius=1, start angle=180, end angle= 360];
\draw [thick] (-2.5,1.2) -- (-.5,1.2);
\draw [thick] (2.5,1.2) -- (.5,1.2);
\draw [thick] (.5,1.2) arc [radius=1, start angle=180, end angle=0];
\draw (1.5,0) -- (1.5,1.2);
\draw (-1.5, 1.2) -- (-1.5,2.4);
\node at (-1.5,.8) {$\sigma$};
\node at (1.5,1.6) {$Y$};
\end{tikzpicture}
\end{center}

Finally we assume, just as it is done in \cite{Henson2014}, that there exists a unique way of discarding a system, which we denote by

\begin{center}
\begin{tikzpicture}
\earth{0}{1.2}
\draw (0,0) -- (0,1.2);
\node at (1,0) {.};
\end{tikzpicture}
\end{center}

We call this the \emph{discarding test}, and it is shown in \cite{Chiribella2010} that its existence and uniqueness is equivalent to the non-signaling condition.
These elements can now be connected by using the output system of one test as input system for another. An arrangment of tests is called a \emph{generalized Bayesian network} (GBN). We also say that the arrangement forms a GBN \emph{with respect to} a DAG $\mathcal{G}$, or that a GBN has \emph{shape} $\mathcal{G}$, if the tests are arranged according to it, analogous to the classical case introduced in the main text.

The main ingredient for the proof of \eqref{theineq} for $n=3$ in \cite{Henson2014} is the following
\begin{lem}[\cite{Henson2014}, Thm. 23.]\label{cutfree32}
For any probability distribution $p(x,y,z)$ of random variables that are the classical output of a GBN with respect to the DAG in Figure \ref{gtri} there is a probability distribution $p'$ such that
\begin{eqnarray}
 p'(x,z)&=&p'(x)p'(z)\\
 p'(x,y)&=&p(x,y)\\
 p'(y,z)&=&p(y,z)
\end{eqnarray}
\end{lem}

For our purposes we need a generalization of this result. The GBN for the scenario of any parent connecting at most $m$ children can be described as follows. The $n$ random variables $V_1,...,V_n$ arise from $n$ measurement tests. For any $m$ of these measurement tests there is a preparation test whose output systems are input for exactly these measurements. We denote the preparation test corresponding to a subset $I\subset \{1,...,n\},\ |I|=m$ by $\sigma_{I}$. In total there are therefore $n \choose m$ preparations. This GBN can be found in Figure \ref{gnm}, and we denote the corresponding DAG by $\mathcal{G}_{n,m}$.

\begin{widetext}
 \begin{center}
 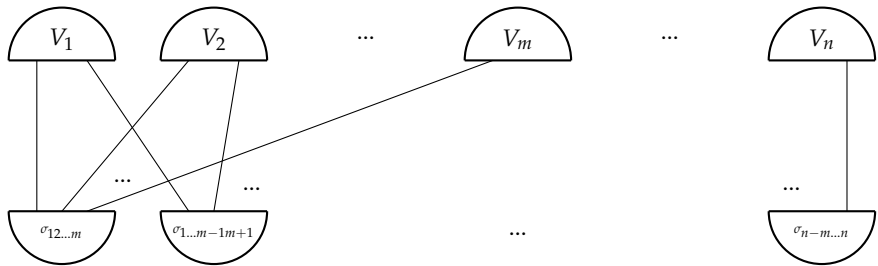
\begin{figure}[!h]
\begin{tikzpicture}
 \prep{-5}{0}{\tiny$\sigma_{12...m}$}
 \prep{-3}{0}{\tiny$\sigma_{1...m-1m+1}$}
 \node at (1,-.3)  {...};
 \prep{5}{0}{\tiny$\sigma_{n-m...n}$}
 \meas{-5}{2}{$V_1$}
  \meas{-3}{2}{$V_2$}
  \node at (-1,2.3) {...};
   \meas{1}{2}{$V_m$}
    \node at (3,2.3) {...};
    \meas{5}{2}{$V_n$}

   \draw (-5.33,0) -- (-5.33,2);
   \draw (-5,0) -- (-3.33,2);
   \node at (-4.2,.4) {...};
   \draw (-4.67,0) -- (0.67,2);
    \draw (-3.33,0) -- (-4.67,2);
   \draw (-3.0,0) -- (-2.67,2);
   \node at (-2.5,.3) {...};
   \node at (4.6,.3) {...};
   \draw (5.33,0) -- (5.33,2);

\end{tikzpicture}

\caption{A generalized Bayesian network for $\mathcal{G}_{n,m}$}\label{gnm}
\end{figure}

\end{center}
\end{widetext}

\begin{figure}
\begin{center}

\pgfmathsetmacro{\t}{1.}
\pgfmathsetmacro{\s}{2.5}
\begin{tikzpicture}

\prep{\t *-3}{0}{\small$\sigma_{12}$}
\prep{\t *-1}{0}{\small$\sigma_{13}$}
\prep{\t *1}{0}{\small$\sigma_{13}'$}
\prep{\t *3}{0}{\small$\sigma'_{23}$}

 \meas{\t *-2}{\s}{$V_2$}
  \meas{\t*0}{\s}{$V_3$}
   \meas{\t *2}{\s}{$V_4$}

  \draw (\t *-3.3,0) -- (\t *-2.3,\s);
  \draw (\t *-2.7,0) -- (\t *-0.3,\s);
 \draw (\t *-1.3,0) -- (\t *-1.7,\s);

  \draw (\t *-0.7,0) -- (\t *-0.6,.3*\s);
  \draw (\t *0.7,0) -- (\t *0.6,.3*\s);
 \draw (\t *1.3,0) -- (\t *1.7,\s);

  \draw (\t *2.7,0) -- (\t *.3,\s);
   \draw (\t *3.3,0) -- (\t *2.3,\s);

     \earth{\t *-.6}{.3*\s}
	\earth{\t *.6}{.3*\s}
\end{tikzpicture}

\end{center}
\caption{An example of the modified GBN from Lemma \ref{cutfree} for $n=3$ and $i=2$}\label{g34}
\end{figure}
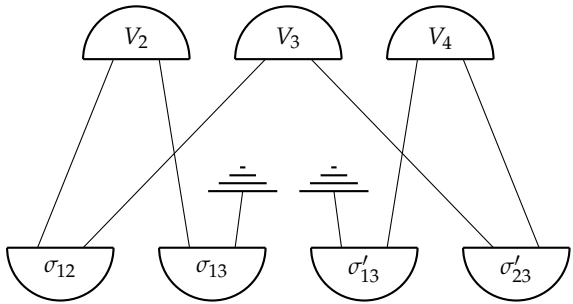

\begin{lem}\label{cutfree}
 For any probability distribution arising from a GBN of shape $\mathcal{G}_{n,m}$ and any index $i\in \{1,...,n\}$ there is a probability distribution $p'$ such that any bivariate marginal involving $V_i$ is equal to the corresponding marginal of $p$ and $p'(v_1,...v_{i-1},v_{i+1},...,v_n)$ is compatible with $\mathcal{G}_{n-1,m-1}$.
\end{lem}
\begin{proof}
 Analoguous to the proof of Lemma \ref{cutfree32} in \cite{Henson2014} we define a new GBN from the old one as follows:
 \begin{itemize}
  \item Any preparation test $\sigma_I$ with $i\in I$ is left as it is.
  \item Any preparation test $\sigma_I$ with $i\notin I$ is copied. In one copy the first outgoing edge is discarded and in the second copy all edges except the first are discarded.
 \end{itemize}
 The modified GBN is depicted in Figure \ref{g34} for $n=3$, $m=2$ and $i=2$.
 It can be seen in an analogous way as in the proof of Theorem 23 in \cite{Henson2014} that using the probability distribution that arises from this GBN as $p'$ has the desired properties.
\end{proof}

We are now ready to prove the inequality \eqref{theineq}.

\begin{thm}
Let $V_1,...,V_n$ be random variables defined by a GBN of shape $\mathcal{G}_{n,2}$. Then for any $i\in\{1,...,n\}$
\begin{equation}
 \sum_{\substack{j\in\{1,...,n\}\\ j\neq i}} I(V_i:V_j)\le H(V_i)
\end{equation}
\end{thm}

\begin{proof}
Without loss of generality we take $i=1$.
We proceed by induction over $n$. For $n=2$ the inequality is trivially true. Assume now that the statement is true for $n-1$. We construct the probability distribution $p'$ according to Lemma \ref{cutfree} and observe that $p'(v_2,...,v_n)$ arises from $\mathcal{G}_{1,n-1}$, i.e. it is a product distribution. Denote the modified random variables by $V_1',...,V_n'$ and calculate
\begin{widetext}
\begin{eqnarray}
 \sum_{j=2}^n I(V_1:V_j)&=& \sum_{j=2}^n I(V'_1:V'_j)\nonumber\\
 &=&I(V'_2:V'_3)-I(V'_2:V'_3|V'_1)+I(V'_1:V'_2V'_3)+\sum_{j=4}^n I(V'_1:V'_j)\nonumber\\
 &\le&I(V'_1:V'_2V'_3)+\sum_{j=4}^n I(V'_1:V'_j),
\end{eqnarray}
\end{widetext}
where the inequality follows from the independence of $V_2'$ and $V_3'$ and srong subadditivity.
Now observe that with $X=(V'_2,V'_3)$ the distribution of $V'_1, X, V'_4,...,V'_n$ is compatible with $\mathcal{G}_{n-1,2}$ and therefore we have, using the induction hypothesis,
\begin{equation}
I(V'_1:V'_2V'_3)+\sum_{j=4}^n I(V'_1:V'_j)\le H(V_1').
\end{equation}
But $p'(v_1)=p(v_1)$ and therefore \eqref{theineq} is proven.
\end{proof}

For general $m$ the situation is somewhat less simple, but for the special case of $n=m+1$ we can still prove a nontrivial inequality.

\begin{thm}
 Let $V_1,...,V_{m+1}$ be random variables corresponding to a GBN of shape $\mathcal{G}_{m+1,m}$. Then
 \begin{equation}\label{eq:nequalsmp1}
  \sum_{k=2}^{m+1}I(V_1:V_k)\le \sum_{k=0}^{m-3}\frac{(m-k-1)(m-k-1)!}{(m-1)!}H(V_{k+1}),
 \end{equation}
 and there is a set of random variables $X_1,...,X_{m+1}$ incompatible with $\mathcal{G}_{m+1,m}$ that violates this inequality.
\end{thm}

Note that this inequality is, in particular, also true for quantum-classical bayesian networks and, to our knowledge, provides the only known entropic corollaries in this case, too.
\begin{proof}(by induction)
For $m=1$ the statement is trivially true, as then the two random variables are independent and therefore $I(V_1:V_2)=0\le 0$. Assume now the inequality was proven for $m-1$. Construct random variables $V_1',...,V_{m+1}'$ according to Lemma \ref{cutfree}. Then calculate
\begin{widetext}
\begin{eqnarray}
 \sum_{k=2}^{m+1}I(V_1:V_k)&=&\sum_{k=2}^{m+1}I(V'_1:V'_k)\nonumber\\
 &=&\frac{1}{m-1}\sum_{k=3}^{m+1}\left[I(V'_2:V'_k)-I(V'_2:V'_k|V'_1)+I(V'_1:V'_2V'_k)+(m-2)I(V'_1:V'_k)\right]\nonumber\\
 &\le&\frac{1}{m-1}\sum_{k=3}^{m+1}\left[I(V'_2:V'_k)+I(V'_1:V'_2V'_k)+(m-2)I(V'_1:V'_k)\right]\nonumber\\
 &\le&(m-1)H(V'_1)+\frac{1}{m-1}\sum_{k=0}^{m-4}\frac{(m-k-2)(m-k-2)!}{(m-2)!}H(V_{k+2})\nonumber\\
 &=&(m-1)H(V'_1)+\sum_{k=1}^{m-3}\frac{(m-k-1)(m-k-1)!}{(m-1)!}H(V_{k+1})\nonumber\\
 &=&\sum_{k=0}^{m-3}\frac{(m-k-1)(m-k-1)!}{(m-1)!}H(V_{k+1}),
\end{eqnarray}
\end{widetext}
where the first inequality follows from strong subadditivity and the second inequality follows from the induction hypothesis and the trivial bound $I(X:Y)\le H(X)$. This completes the proof of the first assertion, i.e. that the inequality is fulfilled by random variables from a GBN of shape $\mathcal{G}_{m+1,m}$. To see that more general random variables violate this inequality, let $X_i=X, i=1,...,m+1$, where $X$ is an unbiased coin. In other words, the $X_i$ are maximally corellated. Then $H(X_i)=1$ and $I(X_1:X_i)=1$. Therefore we have
\begin{equation}
 \sum_{k=2}^{m+1}I(X_1:X_k)=m,
\end{equation}
but
\begin{widetext}
\begin{eqnarray}
 \sum_{k=0}^{m-3}\frac{(m-k-1)(m-k-1)!}{(m-1)!}H(X_{k+1})&=&\sum_{k=0}^{m-3}\frac{(m-k-1)(m-k-1)!}{(m-1)!}\nonumber\\
 &=&m-\frac{2}{(m-1)!}<m,
\end{eqnarray}
\end{widetext}
hence the inequality \eqref{eq:nequalsmp1} is violated
\end{proof}

Note that this inequality yields nontrivial constraints for the entropies of random variables resulting from a GBN of any shape $\mathcal{G}_{n,m},\ n>m$, as it can be applied to any $m+1$ of the $n$ variables.

\bibliography{QDAGsBib}

\end{document}